\documentclass[final]{colt2022-test} 

\title[Improved Parallel Algorithm for Minimum Cost Submodular Cover Problem]{Improved Parallel Algorithm for Minimum Cost Submodular Cover Problem}
\usepackage{times}
\usepackage{indentfirst}
\usepackage{algorithm,algpseudocode,color,algcompatible}
\usepackage{hyperref}
\usepackage{pstricks}
\usepackage{tikz}
\usepackage{caption}
\usepackage{natbib}

 \coltauthor{\Name{Yingli Ran} \Email{ranyingli@zjnu.edu.cn}\and
  \Name{Zhao Zhang}\thanks{Corresponding author} \Email{hxhzz@sina.com}\\
  \addr College of Mathematics and Computer Sciences, Zhejiang Normal University
   \AND
 \Name{Shaojie Tang} \Email{shaojie.tang@utdallas.edu}\\
 \addr Naveen Jindal School of Management, University of Texas at Dallas
}

\begin{document}

\maketitle

\begin{abstract}
In the minimum cost submodular cover problem (MinSMC), we are given a monotone nondecreasing submodular function $f\colon 2^V \rightarrow \mathbb{Z}^+$, a linear cost function $c: V\rightarrow \mathbb R^{+}$, and an integer $k\leq f(V)$, the goal is to find a subset $A\subseteq V$ with the minimum cost such that $f(A)\geq k$. The MinSMC  can be found at the heart of many machine learning and data mining applications. In this paper, we design a parallel algorithm for the MinSMC  that takes at most $O(\frac{\log km\log k(\log m+\log\log mk)}{\varepsilon^4})$ adaptive rounds, and it achieves an approximation ratio of $\frac{H(\min\{\Delta,k\})}{1-5\varepsilon}$ with probability at least  $1-3\varepsilon$, where $\Delta=\max_{v\in V}f(v)$, $H(\cdot)$ is the Harmonic number, $m=|V|$, and $\varepsilon$ is a constant in $(0,\frac{1}{5})$. 
\end{abstract}

\begin{keywords}
minimum cost submodular cover; approximation algorithm; parallel algorithm.
\end{keywords}
\section{Introduction}
Recently, submodular optimization has attracted considerable interest in machine learning and data mining, and optimizing a submodular function can be found in a variety of applications, including viral marketing \citep{Kempe}, information gathering \citep{Krause}, and active learning \citep{Golovin}. In this paper, we design parallel approximation algorithms for the fundamental problem of {\em minimum cost submodular cover} (MinSMC). The input of MinSMC is a set $V$ of $m$ elements. Given a monotonically nondecreasing submodular function $f\colon 2^V \rightarrow \mathbb{Z}^+$, a linear cost function $c: V\rightarrow \mathbb R^{+}$, and an integer $k\leq f(V)$, the goal of the MinSMC  is to find a subset $A\subseteq V$ with the minimum cost such that $f(A)\geq k$, where the cost of $A$ is $c(A)=\sum\limits_{v\in A}c(v)$. The MinSMC has numerous applications, including data summarization \citep{Tschiatschek} and recommender systems \citep{El-Arini}. In the example of data summarization, we are given a set of data, our goal is to select a cheapest set of data, whose representativeness meets some minimum requirement. Many commonly used utility functions exhibit submodularity, a natural diminishing returns property, leading to the formulation of a  MinSMC \citep{Mirzasoleiman1}. To solve MinSMC,  \cite{Wolsey} developed a centralized sequential greedy algorithm which an approximation ratio of $H(\Delta)$, where $H(\Delta)=\sum_{i=1}^{\Delta}1/i$ is the $\Delta$-th Harmonic number and $\Delta=\max_{v\in V}f(v)$.

Unfortunately, the aforementioned centralized sequential greedy method requires $\Omega(n)$ adaptive rounds. A formal definition of \emph{adaptive round} is presented in Definition \ref{def:1}. To this end, we are interested  designing effective parallel algorithms for MinSMC.  The state-of-the-art  parallel algorithm for the {\em unweighted} MinSMC was presented in \citep{Fahrbach}; it takes at most $O(\log(m\log k)\log k)$ adaptive rounds to produce a solution sized at most $O(\log k|OPT|)$, where $OPT$ is the optimal solution.  
Note that their approximation ratio is dependent on $k$ which might be as large as $\Theta(n)$,  whereas  $\Delta$ might be much smaller than $k$. In this work we consider a general \emph{weighted} MinSMC and we develop an effective and efficient parallel algorithm,  whose approximation ratio is arbitrarily close to $H(\Delta)$.

\subsection{Related Works}

 Recently,  \citet{Balkanski} introduced the concept of ``adaptive complexity'' which is defined as the  number of parallel rounds required to achieve a constant factor approximation ratio. We use the same notation to measure the running time of a parallel algorithm. \cite{Chekuri} describe parallel algorithms for approximately maximizing the multilinear relaxation of a monotone submodular
function subject to packing constraints. Both aforementioned studies focus on developing effective parallel algorithms for constrained submodular maximization problem, whereas we study the MinSMC. For the MinSMC, Wolsey \citep{Wolsey} presented a greedy algorithm that achieves an approximation ratio of $H(\Delta)$. For the \emph{unweighted} version of the MinSMC, \citet{Fahrbach} developed a parallel algorithm whose approximation ratio is at most of $O(\log k)$ and it takes $O(\log(m\log k)\log k)$ adaptive rounds. For the set cover problem (finding the smallest subcollection of sets that covers all elements), a special case of the MinSMC, \citet{Berger} provided the first parallel algorithm whose approximation guarantee is similar to that of the centralized greedy algorithm. They used the bucketing technique to obtain a $(1+\varepsilon)H(n)$-approximation in $O(\log^5 M)$ rounds, where $M$ is the total sum of the sets' sizes.  \citet{Rajagopalan} improved the number of rounds to $O(\log^3(Mn))$ at the cost of a larger approximation ratio of $2(1+\varepsilon)H(n)$. \citet{Blelloch} further enhanced these results by obtaining a $(1+\varepsilon)H(n)$-approximation algorithm in $O(\log^2 M)$ rounds.  We list the performance bounds of the closely related studies in Table \ref{tab1}.

\begin{table}[htpb]\footnotesize
\centering
\begin{tabular}{c cc}
\hline
\textbf{Source} & \textbf{Approximation ratio}& \textbf{\# of adaptive rounds}\\
\hline
\underline{Minimum submodular cover} &  \\
\hline
\textbf{Our algorithm} & $\frac{H(\min\{\Delta,k\})}{1-5\varepsilon}$ & $O(\frac{\log km\log k(\log m+\log\log mk)}{\varepsilon^4})$\\
\hline
\citep{Wolsey}(sequential) & $H(\Delta)$ & $\Omega(m)$ \\
\hline
\citep{Fahrbach}(unweighted MinSMC) & $O(\log k)$ & $O(\log(m\log k)\log k)$\\
\hline
\underline{Minimum set cover} &  \\
\hline
\citep{Berger} & $(1+\varepsilon)H(n)$ & $O(\log^5 M)$\\
\hline
\citep{Rajagopalan} & $2(1+\varepsilon)H(n)$ & $O(\log^3(Mn))$\\
\hline
\citep{Blelloch} & $(1+\varepsilon)H(n)$ & $O(\log^2 M)$\\
\hline
\end{tabular}
  \caption{Performance bounds of the closely related studies.}\label{tab1}
\end{table}


\subsection{Our contributions and technical overview}
In this paper, we design a parallel algorithm for the MinSMC. Our algorithm achieves a near-optimal $\frac{H(\min\{\Delta,k\})}{1-5\varepsilon}$-approximation with probability of at least $1-3\varepsilon$, and it takes  poly-logarithmic $O(\frac{\log km\log k(\log m+\log\log mk)}{\varepsilon^4})$ adaptive rounds, where $\varepsilon$ is a constant in $(0,\frac{1}{5})$. 
Note that the aforementioned $O(\log k)$-approximation algorithm \citep{Fahrbach} only works for the unweighted  MinSMC, and $\Delta$ might be much smaller than $k$. 
One naive approach to solve our problem is to iteratively call a parallel algorithm for the submodular maximization problem with a knapsack constraint \citep{Chekuri} until we find a feasible solution to the MinSMC. Unfortunately, the approximation ratio of this method is $O(\log k)$. Nonetheless, further effort is needed to improve this ratio to  $O(\log \Delta)$.
We build our algorithm through novel combinations of the  ideas of multilayer bucket \citep{Berger}, maximal nearly independent set \citep{Blelloch}, and random sample \citep{Fahrbach}. Note that both \citet{Berger} and \citet{Blelloch} focus on the set cover problem, which is a special case of the MinSMC. Their approach can not be applied to solving the MinSMC directly. When applied separately, both of them encounter some structural difficulties.

\textbf{High-Level Intuition of Our Proposed Approach:} Inspired by the parallel algorithm \citep{Blelloch} for the minimum set cover problem, we design a parallel mechanism to imitate the sequential greedy algorithm. In each iteration of our algorithm, we only consider those elements with similar marginal profit-to-cost ratios. We find a nearly independent set from those elements such that the  profit-to-cost ratio of the  nearly  independent set  is almost the same as that of a best single element. Unfortunately, it is not clear how to adapt the solution from \citep{Blelloch} to find a  nearly independent set in the context of MinSMC. To overcome this challenge, we are inspired by \citep{Fahrbach} to use a randomized selection and guessing technique. Note that the original solution developed in \citep{Fahrbach} only works for the unweighted MinSMC. To deal with the weighted case, we further group elements into buckets such that all elements from the same bucket have similar profit-to-cost ratios and marginal profits. In this way, all elements from the same bucket have similar costs. This idea of a multi-layer-bucket originated in \citep{Berger} for solving the minimum set cover problem, we extend their idea to the general MinSMC. 

The remaining part of this paper is organized as follows. The design of our parallel algorithm  and its analysis are presented in Section \ref{sec2} and \ref{sec3}. Section \ref{sec3} concludes the paper with some discussions on future work.
\section{Parallel Algorithm and Analysis for MinSMC}\label{sec2}

\subsection{Preliminaries}\label{subsec2.1}

\begin{definition}[Submodular and monotone nondecreasing function]
Given a set $V$ consisting of elements $v_1,v_2\ldots,v_m$, and a function $f\colon 2^V \rightarrow \mathbb{R}^+$, $f$ is {\em submodular} if $f(A)+f(B)\geq f(A\cap B)+f(A\cup B)$ for any $A,B\subseteq V$; $f$ is {\em monotone non-decreasing} if $f(A)\geq f(B)$ for any $B\subseteq A \subseteq V$.
\end{definition}

For any two sets $A,B\subseteq V$, denote $f_A(B)=f(A\cup B)-f(A)$ to be the {\em marginal profit} of $B$ over $A$. Assume $f(V)=n$. In this paper, $f$ is always assumed to be an integer-valued, monotone nondecreasing, submodular function. It can be verified that for any $A\subseteq V$, the marginal profit function $f_A(\cdot)$ is also a monotone nondecreasing, submodular function.

\begin{definition}[Minimum Submodular Cover Problem (MinSMC)]
{\rm Given a monotone nondecreasing integer submodular function $f\colon 2^V \rightarrow \mathbb{Z}^+$, a cost function $c: V\rightarrow \mathbb R^{+}$, an integer $k\leq n$, MinSMC can be formulated as follows:
\begin{equation}\label{eq0625-1}
\min\{c(A)\colon A\subseteq V, f(A)\geq k\},
\end{equation}
where $c(A)=\sum\limits_{v\in A}c(v)$.}
\end{definition}
Define a function $g$ as $g(A)=\min\{f(A),k\}$ for any subset $A\subseteq V$. When $f$ is a monotone nondecreasing submodular function, it can be verified that $g$ is also a monotone nondecreasing submodular function. Note that $\max\{g(A)\colon A\subseteq V\}=k=g(V)$, and for the modified MinSMC
\begin{equation}\label{eq0625-2}
\min\{c(A)\colon g(A)=g(V)\},
\end{equation}
a set $A$ is feasible to \eqref{eq0625-2} if and only if $A$ is feasible to \eqref{eq0625-1}. Hence, problems \eqref{eq0625-1} and \eqref{eq0625-2} are equivalent in terms of approximability, that is, $A$ is an $\alpha$-approximate solution to problem \eqref{eq0625-2} if and only if $A$ is an $\alpha$-approximate solution to problem \eqref{eq0625-1}. In the following, we concentrate on the modified MinSMC \eqref{eq0625-2}.

We next introduce the notation of  \emph{$\varepsilon$-nearly independent set}, which is adapted from the concept of {\em $\varepsilon$-maximal nearly independent set} ($\varepsilon$-MaxNIS) \citep{Blelloch}. 

\begin{definition}[$\varepsilon$-nearly independent set ($\varepsilon$-NIS)]\label{def0703-1}
{\rm For a real number $\varepsilon>0$ and a set $S\subseteq V$, we say that a set $J\subseteq V\setminus S$ is an $\varepsilon$-NIS with respect to $S$ and $\varepsilon$ if $J$ satisfies the following nearly independent property:
\begin{equation}\label{eq0709-1}
g_{S}(J)\geq (1-\varepsilon)^2\sum_{v\in J}g_{S}(v).
\end{equation}}

\end{definition}

At last, we introduce the notation of \emph{adaptive round} from \citep{Balkanski}. For simplicity, we use \emph{adaptive round} and \emph{round} interchangeably in the rest of this paper.
\begin{definition}[Adaptive round]
 Given a value oracle $f$, which  receives a set $S \subseteq N$ and returns its value $f(S)$, we say an algorithm takes $\gamma$ adaptive rounds if (1) every query to $f$ in round $i \in [\gamma]$ depends only on the answers to queries in rounds $1$ to $i-1$, and (2) it performs polynomially-many parallel queries in each adaptive round.
\label{def:1}
\end{definition}

\subsection{Outline of Our Algorithm}

We first present an outline of our algorithm. The main algorithm is presented in Algorithm \ref{algo1}, which can be viewed as a parallel implementation of the sequential greedy algorithm \citep{Wolsey}. The overall design of our algorithm follows the framework of adaptive greedy cover algorithm presented in \citep{Fahrbach} for the unweighted MinSMC. However, we need to design a new algorithm and  perform more sophisticated analysis to achieve a better approximation ratio  for the general MinSMC.
 Our basic idea is to process elements in batches in accordance with their profit-to-cost ratio and profits. Specifically, in each iteration of our algorithm, we first collect a group of elements with similar profit-to-cost ratio and profits, then we pick an $\varepsilon$-NIS from them using Algorithm \ref{algo2} and add it to the solution. The main idea of Algorithm \ref{algo2} is to guess the size of a largest $\varepsilon$-NIS, and for each guess, we use a mean function (Algorithm \ref{algo3}) to verify whether it is an $\varepsilon$-NIS or not. Let $c_{\max}$ and $c_{\min}$ be the maximum and the minimum cost of elements, respectively.
 Unfortunately, the running time of  Algorithm \ref{algo1} is dependent on $c_{\max}/c_{\min}$, whose value could be arbitrarily large.  In order to achieve the running time of $O(\frac{\log mk\log k\log (m\log mk)}{\varepsilon^4})$, we add a preprocessing stage to Algorithm \ref{algo1} to ensure that $c_{\max}/c_{\min}$ is upper bounded by $O(\log mk)$.  Our final algorithm is presented in Algorithm \ref{algo4}. 



\begin{algorithm}[t]
\caption{MinSMC-Par$(V,g,c,k,\varepsilon)$}\label{algo1}
\textbf{Input:} MinSMC-Par instance $(V,g,c,k,\varepsilon)$.

\textbf{Output:} A subset $B\subseteq V$ with $g(B)\geq k$.
	
\begin{algorithmic}[1]
\STATE $t\leftarrow 1$
\STATE $\beta=\max_{v\in V}g(v)/c(v)$\label{line2}
\STATE $\tau=\max_{v\in V}g(v)$\label{line3}
\STATE $T=\log_{1/(1-\varepsilon)}kc_{\max}/c_{\min}$
\STATE $\ell=\log_{1/(1-\varepsilon)}k$
\STATE $t'\leftarrow 1$
\STATE $B_{1}^{1}\leftarrow \emptyset$
\WHILE {$t\leq T$}
\WHILE {$t'\leq \ell$}
\STATE $A_{t}^{t'}=\left\{v\in V
\colon \begin{array}{l}(1-\varepsilon)^t\beta \leq g_{B_{t}^{t'}}(v)/c(v)\leq (1-\varepsilon)^{t-1}\beta,\\ (1-\varepsilon)^{t'}\tau\leq g_{B_{t}^{t'}}(v)\leq (1-\varepsilon)^{t'-1}\tau\end{array}\right\}$\label{line11}
\IF {$A_{t}^{t'}=\emptyset$}
\STATE break and go to line \ref{line19}
\ENDIF
\STATE $J_{t}^{t'}\leftarrow$ NIS$(A_{t}^{t'},B_{t}^{t'},\varepsilon,\ell,(1-\varepsilon)^{t-1}\beta,(1-\varepsilon)^{t'-1}\tau)$\label{line0704-1}
\STATE $B_{t}^{t'+1}\leftarrow B_{t}^{t'}\cup J_{t}^{t'}$\label{line12}
\STATE $t'\leftarrow t'+1$
\ENDWHILE
\IF {$g(B_{t}^{t'})\geq k$}\label{line19}
\STATE break and go to line \ref{line20}
\ENDIF
\STATE $B_{t+1}^{1}\leftarrow B_{t}^{t'}$
\STATE $t\leftarrow t+1$
\ENDWHILE
\STATE return $B\leftarrow B_{t}^{t'}$\label{line20}

\end{algorithmic}
\end{algorithm}

\begin{algorithm}[t]
\caption{NIS$(A,B,\varepsilon,\ell, \beta,\tau)$}\label{algo2}
\textbf{Input:} Two sets $A$ and $B$, two threshold values $\beta$ and $\tau$, a constant $0<\varepsilon<1/4$, a parameter
$\ell$.

\textbf{Output:} An $\varepsilon$-nearly independent set $J\subseteq A$ with respect to $B$.
	
\begin{algorithmic}[1]
\STATE $J_1\leftarrow \emptyset$
\STATE $p\leftarrow 1$
\STATE $i\leftarrow -1$
\STATE $A_0\leftarrow A$
\STATE $B_1\leftarrow B$\label{line5}
\STATE $\bar{\varepsilon}\leftarrow \frac{1}{3}(1-\frac{1}{2T\ell})\varepsilon$
\STATE $r\leftarrow \log_{\frac{1}{1-\bar{\varepsilon}}}(2mT\ell)/\varepsilon$
\STATE $\delta\leftarrow \varepsilon/(2rkT^2\ell)$
\WHILE {$p\leq r$}\label{line9}
    \STATE $A_{p}\leftarrow \{v\in A_{p-1}\colon (1-\varepsilon)\beta\leq g_{B_p}(v)/c(v)\leq \beta, (1-\varepsilon)\tau\leq g_{B_p}(v)\leq \tau\}$\label{line6}
    \IF {$A_{p}\leftarrow \emptyset$}
    \STATE break (exit the while loop)
    \ENDIF
    \FOR {$i\leq \log_{1+\bar{\varepsilon}}m$}\label{line14}
    \STATE $t_p\leftarrow \min\{\lfloor(1+\bar{\varepsilon})^i\rfloor,|A_{p}|\}$\label{line15}
    \STATE $\bar{\mu}_p\leftarrow$Mean($B_p,A_p,t_p,\tau,\bar{\varepsilon},\delta$)\label{line17}
    \IF {$\bar{\mu}_p\leq 1-1.5\bar{\varepsilon}$}\label{line18}
    \STATE break (exit the for loop) \label{line181}
    \ENDIF
    \STATE $i\leftarrow i+1$
    \ENDFOR
    \STATE select a $t_p$-set $T_p$ from $A_p$ uniformly at random, and let $J_{p+1}\leftarrow J_{p}\cup T_{p}$\label{line23}
    \STATE $B_{p+1}\leftarrow B_{p}\cup J_{p+1}$\label{line24}
     \IF {$g(B_{p+1})\geq k$}\label{line25}
    \STATE break (exit the while loop)
    \ENDIF
    \STATE $p\leftarrow p+1$
\ENDWHILE
\STATE Return $J_p$\label{line30}
\end{algorithmic}
\end{algorithm}
\subsection{Design of Algorithm \ref{algo1}}\label{subsec2.2}

Now we are ready to present the details of  Algorithm \ref{algo1}. We defer the description of the final Algorithm \ref{algo4} to the next section. 
Let $\beta= \max_{v\in V}g(v)/c(v)$ denote the largest profit-to-cost ratio of a single element. In each round (line \ref{line11} of Algorithm \ref{algo1}), 
we construct a bucket  $A_t^{t'}$ such that all elements in $A_t^{t'}$ have  similar marginal profits and marginal profit-to-cost ratio, e.g., \[A_{t}^{t'}=\left\{v\in V
\colon \begin{array}{l}(1-\varepsilon)^t\beta \leq g_{B_{t}^{t'}}(v)/c(v)\leq (1-\varepsilon)^{t-1}\beta,\\ (1-\varepsilon)^{t'}\tau\leq g_{B_{t}^{t'}}(v)\leq (1-\varepsilon)^{t'-1}\tau\end{array}\right\},\] where $B_{t}^{t'}$ denotes the set of already selected elements before this round, and $\tau=\max_{v\in V}g(v)$ is the largest profit of a single element; then it picks an $\varepsilon$-NIS with respect to $B_{t}^{t'}$  from $A_{t}^{t'}$ using Algorithm \ref{algo2} (see line \ref{line0704-1} and line \ref{line12} of Algorithm \ref{algo1}) and adds it to the solution. A detailed description of Algorithm \ref{algo2} will be provided in the next paragraph.
 In the process of selecting elements, we give higher priority to those elements with higher profit-to-cost ratio. For those elements with the same profit-to-cost ratio, we give higher priority to those with higher marginal profit. We can prove that when the algorithm terminates, with high probability, it outputs a feasible solution with a good approximation.

We next explain Algorithm \ref{algo2} in details. Given two sets $A$ and $B$, a constant $0<\varepsilon<1/4$, the goal of  Algorithm \ref{algo2} is to compute an $\varepsilon$-NIS with respect to $B$ from $A$. For ease of presentation, we call a set consisting of $t$ elements  {\em $t$-set}.  Starting with, $p=1$, $B_1=B$ and $A_0=A$. In the $p$-th round of the  while loop  (line \ref{line9}), we compute a ``good'' $\varepsilon$-NIS with respect to $B_p$ from $A_p$, where $B_p$ is the set of selected elements before round $p$ and $A_p$ is defined in line \ref{line6}. Then we add this  $\varepsilon$-NIS to $B_p$ to obtain $B_{p+1}$. This process takes at most $r$ rounds, where $r= \log_{\frac{1}{1-\bar{\varepsilon}}}(2mT\ell)/\varepsilon$. 
To compute the $\varepsilon$-NIS in each round $p$, we use  a for loop (line \ref{line14}) to guess its size $t_p$.  
For each guess,  we use Algorithm \ref{algo3} to measure the expected quality of a $t_p$-set that is sampled from $A_p$ uniformly at random. Note that we can try all $\log_{1+\bar{\varepsilon}}m$ guesses in parallel.

We next introduce the design of Algorithm \ref{algo3}. We first define a function $I_{t,B,A,\tau,\epsilon}$ as follows. Given two sets $A,B$, a parameter $\tau$, and a real number $0<\varepsilon<1$, for a $t$-set $X$, and an element $x$ from $A\setminus X$, define
$$
I_{t,B,A,\tau,\epsilon}(X,x)=I[g_{B\cup X}(x)\geq (1-\varepsilon)\tau], \mbox{ where $I[\cdot]$ is an indicator function,}
$$
that is, $I_{t,B,A,\tau,\epsilon}(X,x)=1$ if $g_{B\cup X}(x)\geq (1-\varepsilon)\tau$, and $I_{t,B,A,\tau,\epsilon}(X,x)=0$ otherwise. As a convention,
\begin{equation}\label{eq0711-1}
\mbox{if $A\setminus X=\emptyset$, define $I_{t,B,A,\tau,\epsilon}(X,x)=0$.}
\end{equation}

With the above function, Algorithm \ref{algo3} computes $\bar{\mu}_p$, which is the estimated expectation of $I_{t_p,B_p,A_p,\tau,\epsilon}(X,x)$ assuming that  $X$ is a $t_p$-set that is selected from $A_p$ uniformly at random, and $x$ is an element that is drawn uniformly at random from $A_p\setminus X$. It will become clear later that there exists a $t_p$ which ensures that (1) $\bar{\mu}_p\leq 1-1.5\bar{\varepsilon}$, and (2) the random set $T_p$ returned from line \ref{line23} of Algorithm \ref{algo2} is an $\varepsilon$-NIS with respect to $B_p$,  with high probability. 

\begin{algorithm}[htpb]
\caption{Mean($B,A,t,\tau,\bar{\varepsilon},\delta$)}\label{algo3}
\textbf{Input:} ($B,A,t,\tau,\bar{\varepsilon},\delta$).

\textbf{Output:} $\bar{\mu}$.
	
\begin{algorithmic}[1]
\STATE set the number of samples $m'\leftarrow 8\lceil\log(2/\delta)/\bar{\varepsilon}^2\rceil$
\STATE sample $m'$ sets $X_1,\ldots,X_{m'}$ and $m'$ elements $x_1,\ldots,x_{m'}$, where each $X_i$ is a $t$-set selected from $A$ uniformly at random, and $x_i$ is an element sampled from $A\setminus X_i$ uniformly at random
\STATE return $\bar{\mu}\leftarrow \frac{1}{m'}\sum_{i=1}^{m'}I_{t,B,A,\tau,\epsilon}(X_i,x_i)$
\end{algorithmic}
\end{algorithm}

 Unless specified otherwise, we assume $X$ (resp. $X'$) is a $t$-set (resp. $t'$-set) that is selected from $A$ uniformly at random, and $x$ (resp. $x'$) is an element that is drawn uniformly at random from $A\setminus X$ (resp. $A\setminus X'$). We next present a useful lemma to show that $\mathbb{E}_{X, x}[I_{t,B,A,\tau,\epsilon}(X,x)]$ is monotone non-increasing with respect to the sample size $t$.  In the rest of this paper, we will omit the subscript from $\mathbb{E}_{X, x}[\cdot]$ if it is clear from the context, and use a shorthand notation  $I_t(X,x)$ to denote $I_{t,B,A,\tau,\epsilon}(X,x)$.
\begin{lemma}\label{lem0303-2}
Given $B,A,\tau,\epsilon$, suppose $t$ and $t'$ are two integers with $t<t'$. Then
$$
\mathbb{E}[I_t(X,x)]\geq \mathbb{E}[I_{t'}(X',x')].
$$
\end{lemma}

The following lemma reveals the relation between $\bar{\mu}_p$ and $\mathbb{E}[I_{t_p,B_p,A_p,\tau,\epsilon}(X,x)]$.

\begin{lemma}\label{lem0302-1}
With probability at least $1-\delta$, $\mathbb{E}[I_{t_p,B_p,A_p,\tau,\epsilon}(X,x)]\leq 1-\bar{\varepsilon}$ if $\bar{\mu}_p\leq 1-1.5\bar{\varepsilon}$, and $\mathbb{E}[I_{t_p,B_p,A_p,\tau,\epsilon}(X,x)]\geq 1-2\bar{\varepsilon}$ if $\bar{\mu}_p>1-1.5\bar{\varepsilon}$.
\end{lemma}

The proofs of Lemma \ref{lem0303-2} and Lemma \ref{lem0302-1} are given in Appendix A and Appendix B.

\subsection{Performance analysis}\label{subsec2.3}
In this section, we analyze the running time and the approximation ratio of Algorithm \ref{algo1}. We first provide some technical lemmas. The first lemma shows that the expected size of $A_p$ in Algorithm \ref{algo2} decreases exponentially as $p$ grows, which implies that $A_p$ will become empty in at most $\log_{1+\bar\varepsilon}m$ rounds. Note that in line \ref{line14} of Algorithm \ref{algo2}, if $i=\log_{1+\bar{\varepsilon}}m$, then $t_p=|A_p|$, which implies $\bar{\mu}_p=0$ and thus $\bar{\mu}_p\leq 1-1.5\bar{\varepsilon}$. This indicates that line \ref{line181} in Algorithm \ref{algo2} is guaranteed to be triggered. 

\begin{lemma}\label{lem0302}
If line \ref{line181} in Algorithm \ref{algo2} is triggered (the for loop is exited), then $\mathbb{E}[|A_{p+1}|]\leq (1-\varepsilon)|A_p|$ with probability at least $1-\delta$, where $\mathbb{E}[|A_{p+1}|]$ denotes the expected size of $A_{p+1}$ conditioned on a fixed $A_p$.
\end{lemma}
\begin{proof}
The inequality is obvious if $A_{p+1}=\emptyset$. In the following, assume $A_{p+1}\neq\emptyset$.

By the assumption of this lemma, we have $\bar{\mu}_p\leq 1-1.5\bar{\varepsilon}$. Then by Lemma \ref{lem0302-1}, with probability at least $1-\delta$,
\begin{equation}\label{eq0302}
\mathbb{E}[I_{t_p,B_p,A_p,\tau,\epsilon}(X,x)]\leq 1-\bar{\varepsilon}.
\end{equation}
Note that once $T_p$ is picked, for any element $x\in A_p$, we move $x$ to $A_{p+1}$ only if $I[g_{B_p\cup T_p}(x)\geq (1-\varepsilon)\tau, g_{B_p\cup T_p}(x)/c(x)\geq (1-\varepsilon)\beta]=1$; also note that $I[g_{B_p\cup T_p}(x)\geq (1-\varepsilon)\tau, g_{B_p\cup T_p}(x)/c(x)\geq (1-\varepsilon)\beta]=0$ if $x\in T_p$. It follows that
\begin{align*}
\mathbb{E}[|A_{p+1}|] & =\sum_{x\in A_p\setminus T_p}I[g_{B_p\cup T_p}(x)\geq (1-\varepsilon)\tau, g_{B_p\cup T_p}(x)/c(x)\geq (1-\varepsilon)\beta]\\ & \leq \sum_{x\in A_p\setminus T_p}I[g_{B_p\cup T_p}(x)\geq (1-\varepsilon)\tau].
\end{align*}
It follows that
\begin{align*}
\mathbb{E}\left[\frac{|A_{p+1}|}{|A_p\setminus T_p|}\right]=\ & \ \sum_{T_p}\Pr[\textrm{$T_p$ is picked}]\mathbb{E}\left[\frac{|A_{p+1}|}{|A_p\setminus T_p|}|T_p\right]\\
\leq\ & \sum_{T_p}\Pr[\textrm{$T_p$ is picked}]\left(\sum_{x\in A_p\setminus T_p}\Pr[\textrm{$x$ is picked}|T_p]\frac{I[g_{B_p\cup T_p}(x)\geq (1-\varepsilon)\tau]}{|A_p\setminus T_p|}\right)\\
=\ & \ \sum\limits_{T_p,x\in A_p\setminus T_p}\Pr[\textrm{$T_p,x$ are picked}]\frac{I[g_{B_p\cup T_p}(x)\geq (1-\varepsilon)\tau]}{|A_p\setminus T_p|}\\
\leq \ & \ \sum\limits_{T_p,x\in A_p\setminus T_p}\Pr[\textrm{$T_p,x$ are picked}]I[g_{B_p\cup T_p}(x)\geq (1-\varepsilon)\tau]\\ \nonumber
=\ & \ \mathbb{E}[I_{t_p,B_p,A_p,\tau,\epsilon}(T,x)],
\end{align*}
where the second inequality uses the observation that $|A_p\setminus T_p|\geq 1$ (since $A_{p+1}\neq\emptyset$). Combining this with inequality \eqref{eq0302}, we have $\mathbb{E}\left[\frac{|A_{p+1}|}{|A_p\setminus T_p|}\right]\leq 1-\bar{\varepsilon}$. Thus $\mathbb{E}[|A_{p+1}|]\leq (1-\varepsilon)\mathbb{E}[|A_p\setminus T_p|]\leq (1-\bar{\varepsilon})|A_p|$. The lemma is proved.
\end{proof}

For ease of presentation, we call $A_t^{t'}$ (line \ref{line6} of Algorithm \ref{algo2}) as a {\em subordinate bucket} and  $A_{t}=\{v\in V
\colon (1-\varepsilon)^t\beta \leq g_{B_{t}^{t'}}(v)/c(v)\leq (1-\varepsilon)^{t-1}\beta\}$ as a {\em primary bucket}. The following lemma says that for any $t\leq T$ and $t'\leq \ell$, when line \ref{line0704-1} of Algorithm \ref{algo1} returns a set $J_{t}^{t'}$, the subordinate bucket $A_{t}^{t'}$ becomes empty with probability at least $1-\varepsilon/(T^2\ell)$.
\begin{lemma}\label{lem0308-1}
When  Algorithm \ref{algo2} reaches line \ref{line30}, the subordinate bucket $A_p$ (line \ref{line6}) becomes empty with probability at least $1-\varepsilon/(T^2\ell)$.
\end{lemma}

%

The following corollary shows that when the inner while loop of Algorithm \ref{algo1} halts, the primary bucket $A_t$ becomes empty with probability at least $1-\varepsilon/T^2$.

\begin{corollary}\label{cor0308}
After $J_{t}^{\ell}$ is computed (line \ref{line0704-1} of Algorithm \ref{algo1}), the primary bucket $A_t$ becomes empty with probability at least $1-\varepsilon/T^2$.
\end{corollary}

The proofs of Lemma \ref{lem0308-1} and Corollary \ref{cor0308} are given in Appendix C and Appendix D.

The next lemma shows that with probability at least $1-\delta kr$, the set $J_{t}^{t'}$ computed in line \ref{line0704-1} of Algorithm \ref{algo1} satisfies the nearly independent property defined in \eqref{eq0709-1}.

\begin{lemma}\label{lem0303}
$\mathbb{E}[g_{B_{t}^{t'}}(J_{t}^{t'})]\geq (1-\varepsilon)^2\sum_{v\in J_{t}^{t'}}g_{B_{t}^{t'}}(v)$ with probability at least $1-\delta kr$.
\end{lemma}
\begin{proof}
Consider the case when the input $B$ of Algorithm \ref{algo2} is $B_{t}^{t'}$. That is, $B_{1}=B_{t}^{t'}$ (line \ref{line5} of Algorithm \ref{algo2}). We first prove that for any round $p$, with probability at least $1-n\delta$, the random set $T_p$ (line \ref{line23} of Algorithm \ref{algo2}) satisfies
\begin{equation}\label{eq0303-1}
\mathbb{E}[g_{B_p}(T_p)]\geq (1-\varepsilon)^2\sum_{v\in T_p}g_{B_1}(v),
\end{equation}
where $B_p$ is computed in line \ref{line24} of Algorithm \ref{algo2}. We consider a fixed round $p$ in the rest of this proof. For ease of presentation, denote the size of $T_p$ as $t^*$. Inequality \eqref{eq0303-1} is obviously true if $t^*=0$ or $1$. Next, suppose $t^*\geq 2$. Note that line \ref{line23} of Algorithm \ref{algo2} is executed after we jumped out of the for loop. Further note that this jump out is always due to line \ref{line18}. In fact, if the number of iterations the for loop takes has reached $\log_{1+\bar{\varepsilon}}m$, then $t_p=|A_p|$, and every $X_i$ in Algorithm \ref{algo3} is $A_p$, resulting in $\bar{\mu}_p=0$ (see \eqref{eq0711-1}), in which case the condition of line \ref{line18} is satisfied. In the previous round of the for loop, that is, when $t_p$ tries the value $\bar{t}=t^*/(1+\bar\varepsilon)$, we must have $\bar{\mu}_p>1-1.5\bar{\varepsilon}$, and thus
\begin{equation}\label{eq0709-2}
\mathbb{E}[I_{\bar{t},B_p,A_p,\tau,\epsilon}(X,x)]\geq 1-2\bar{\varepsilon}
\end{equation}
by Lemma \ref{lem0302-1}. Assume that $T_p=\{v_1,\ldots,v_{t^*}\}$, and for any $i\leq t^*$, denote $T_{p}^{i}=\{v_1,\ldots,v_i\}$. By the monotonicity of $g$, we have
\begin{equation}\label{eq0303-2}
\mathbb{E}[g_{B_p}(T_p)]\geq \mathbb{E}[g_{B_p}(T_p^{\bar{t}})]=\sum_{i=1}^{\bar{t}}\mathbb{E}[g_{B_p\cup T_{p}^{i-1}}(v_i)].
\end{equation}
By the definition of $I_{i,B_p,A_p,\tau,\epsilon}(X,x)$ and Markov's inequality,
\begin{equation}\label{eq0303-3}
\mathbb{E}[I_{i,B_p,A_p,\tau,\epsilon}(T_{p}^{i},v_{i+1})]=\Pr[g_{B_p\cup T_{p}^{i}}(v_{i+1})\geq (1-\varepsilon)\tau]\leq \frac{\mathbb{E}[g_{B_p\cup T_{p}^{i}}(v_{i+1})]}{(1-\varepsilon)\tau}.
\end{equation}
Combining inequalities \eqref{eq0303-2} and \eqref{eq0303-3}, we have
\begin{equation}\label{eq0308}
\mathbb{E}[g_{B_p}(T_p)]\geq (1-\varepsilon)\tau\cdot\sum_{i=1}^{{\bar{t}}}\mathbb{E}[I_{i,B_p,A_p,\tau,\epsilon}(T_{p}^{i},v_{i+1})].
\end{equation}
For any $i\leq {\bar{t}}$, by Lemma \ref{lem0303-2} and inequality \eqref{eq0709-2}, with probability at least $1-\delta$,
\begin{equation}\label{eq0308-1}
\mathbb{E}[I_{i,B_p,A_p,\tau,\epsilon}(T_{p}^{i},v_{i+1})]\geq 1-2\bar{\varepsilon}.
\end{equation}
Combining inequalities \eqref{eq0308}, \eqref{eq0308-1}, and the union bound, with probability at least $1-n\delta$,
\begin{align}\label{eq0303-4}
\mathbb{E}[g_{B_p}(T_p)]\geq \ & \ (1-2\bar{\varepsilon}){\bar{t}}(1-\varepsilon)\tau\\ \nonumber
= \ & \ \frac{t^*}{1+\bar{\varepsilon}}(1-2\bar{\varepsilon})(1-\varepsilon)\tau\\ \nonumber\geq \ & \ (1-\varepsilon)^2t^*\tau,
\end{align}
where the last inequality is due to  the choice of $\bar{\varepsilon}$. According to line \ref{line11} and line \ref{line0704-1} of Algorithm \ref{algo1}, when Algorithm \ref{algo2} is triggered, we have $g_{B_{1}}(v)\leq \tau$ for any $v\in A$, with respect to the input parameter $\tau$. It follows that $g_{B_{1}}(v)\leq \tau$ holds for every $v\in T_p$. Combining this with \eqref{eq0303-4}, inequality \eqref{eq0303-1} is proved.

Then, by the union bound, and a proof similar to the proof of Corollary \ref{cor0308}, with probability at least $1-\delta kr$,
\begin{equation}\label{eq0706}
\sum_{p=1}^{r}\mathbb{E}[g_{B_p}(T_p)]\geq (1-\varepsilon)^2\sum_{i=p}^{r}\sum_{v\in T_p}g_{B_{1}}(v)
\end{equation}
Combining this with $J_{t}^{t'}=\bigcup_{p=1}^{r}T_p$, with probability at least $1-\delta kr$,
\begin{align*}
\mathbb{E}[g_{B_{t}^{t'}}(J_{t}^{t'})]= \ & \ \sum_{p=1}^{r}\mathbb{E}[g_{B_p}(T_p)]\\ \nonumber \geq \ & \ (1-\varepsilon)^2\sum_{p=1}^{r}\sum_{v\in T_p}g_{B_{1}}(v)\\ \nonumber= \ & \ (1-\varepsilon)^2\sum_{v\in J_{t}^{t'}}g_{B_{t}^{t'}}(v),
\end{align*}
where the last inequality is due to $B_1=B_{t}^{t'}$ and $J_{t}^{t'}=\bigcup_{p=1}^{r}T_p$.
\end{proof}

Without loss of generality, we assume that every inner while loop of Algorithm \ref{algo1} is executed $\ell$ times. Denote by $D_t=J_{t}^{1}\cup J_{t}^{2}\cup \ldots\cup J_{t}^{\ell}$ for any $t\leq T$. The following corollary shows that the expected cost effectiveness of $D_t$ decreases geometrically as $t$ grows.

\begin{corollary}\label{cor0303}
For any $t\leq T$, $\frac{\mathbb{E}[g_{B_{t}^1}(D_t)]}{c(D_t)}\geq (1-\varepsilon)^{t+2}\beta$ with probability at least $1-\delta kr\ell$.
\end{corollary}

The proof of Corollary \ref{cor0303} is presented in Appendix E.

Now, we are ready to analyze the expected performance of Algorithm \ref{algo1}.

\begin{theorem}\label{thm0309}
For any constant $0<\varepsilon<1/4$, with probability at least $1-3\varepsilon$, Algorithm \ref{algo1} outputs an $\frac{H(\min\{\Delta,k\})}{1-4\varepsilon}$-approximate solution to the MinSMC, where $\Delta=\max_{v\in V}f(v)$. It takes at most  $O(\frac{T\log k(\log m+\log (T\log k))}{\varepsilon^3})$ rounds.
\end{theorem}
\begin{proof}
We first analyze the running time of  Algorithm \ref{algo1}.  The two layers of while loops takes at most $T\ell$ iterations, where $T=\log_{1/(1-\varepsilon)}kc_{\max}/c_{\min}$
and $\ell=\log_{1/(1-\varepsilon)}k$. In each iteration, it calls Algorithm \ref{algo2} to find an $\varepsilon$-NIS.  Recall that the for loop in Algorithm \ref{algo2} can be processed in parallel. Moreover, Algorithm \ref{algo3} can also be parallelized using $m'$ parallel queries. It follows that Algorithm \ref{algo2} takes at most $r$ rounds, where $r=\log_{\frac{1}{1-\bar{\varepsilon}}}(2mT\ell)/\varepsilon$. Hence, the running time of Algorithm \ref{algo1} is $O(T\ell r)=O(\frac{T\log k(\log m+\log (T\log k))}{\varepsilon^3})$. A summary of running time analysis is presented in Table \ref{tab2}.
\begin{table}[htpb]
\centering
\begin{tabular}{c c}
\hline
\textbf{Algorithm} & \textbf{\# of adaptive rounds}\\
\hline
Algorithm \ref{algo1} &  $O(T\ell\times {\textrm {number of rounds of Algorithm \ref{algo2}}})$\\
\hline
Algorithm \ref{algo2}  & $O(r\times {\textrm {number of rounds of Algorithm \ref{algo3}}})$ \\
\hline
Algorithm \ref{algo3} &  $O(1)$\\
\hline
\end{tabular}
  \caption{Summary of Running Time Analysis}\label{tab2}
\end{table}

Next we analyze the approximation ratio of Algorithm \ref{algo1}. Let $B$ be the output of Algorithm \ref{algo1}, then $B=D_1\cup \ldots\cup D_T$, $B_t^1=D_1\cup\cdots D_{t-1}$ for $t\geq 2$ and $B_{1}^1=\emptyset$, where $D_t=J_{t}^{1}\cup J_{t}^{2}\cup \ldots\cup J_{t}^{\ell}$ for $t\leq T$. The following claim shows that based on $D_1,D_2,\ldots, D_T$, we can construct a sequence of sets whose expected cost-effectiveness is monotone.

\vskip 0.2cm {\bf Claim 1.} We can construct a sequence of sets $D'_1,D'_2,\ldots,D'_{p}$ with $p\leq T$ such that with probability at least $1-3\varepsilon/2$,
\begin{equation}\label{eq0712-1}
\frac{\mathbb{E}[g_{B_{i}'}(D'_{i+1})]}{c(D'_{i+1})}\leq\frac{\mathbb{E}[g_{B_{i-1}'}(D'_{i})]}{c(D'_{i})}
\end{equation}
holds for any $i\leq p$, where $B_i'=D'_1\cup\ldots\cup D'_i$ for $i\leq p$.

\vskip 0.2cm For a set $B$, denote by $\beta(B)=\max_{v\in V}\frac{g_{B}(v)}{c(v)}$ the maximum marginal profit-to-cost ratio with respect to $B$. The next claim estimates the loss between the expected cost-effectiveness of $D'_i$ and the ratio $\beta(B_{i-1}')$.

\vskip 0.2cm {\bf Claim 2.} For any $1\leq i\leq p-1$, with probability at least $1-3\varepsilon/2$,
$$
\frac{\mathbb{E}[g_{B_{i}'}(D'_{i+1})]}{c(D'_{i+1})}\geq (1-\varepsilon)^4\beta(B_i').
$$


\vskip 0.2cm To complete the estimation of the approximation ratio, we consider an optimal solution $A^*$, and construct an auxiliary weight $w$ as follows. Denote $r_i=\mathbb{E}[g_{B_{i-1}'}(D'_i)]$ and $z_{v,i}=\mathbb{E}[g_{B_{i-1}'}(v)]$ for $1\leq i\leq p$ and $v\in A^*$. For any $v\in A^*$, define
$$
w(v)=\sum\limits_{i=1}^p(z_{v,i}-z_{v,i+1})\frac{c(D'_i)}{r_i},
$$
where $z_{v,p+1}=0$.

\vskip 0.2cm {\bf Claim 3.} With probability at least $1-3\varepsilon/2$, $c(B_p')\leq \sum_{v\in A^*}w(v)$.

\vskip 0.2cm {\bf Claim 4.} With probability at least $1-3\varepsilon/2$, $w(v)\leq c(v)\cdot\frac{H(\min\{\Delta,k\})}{1-4\varepsilon}$.

\vskip 0.2cm Combining Claim 3, Claim 4, and the union bound, with probability at least $1-3\varepsilon$, $c(B_p')\leq \frac{H(\min\{\Delta,k\})}{1-4\varepsilon}c(A^*)$. The approximation ratio is proved.
\end{proof}

The proofs of Claim 1-4 are presented in Appendix $F-I$ respectively.

\section{Completing the Last Piece of the Puzzle: Bounding $c_{\max}/c_{\min}$}
\label{sec3}
Note that the running time in Theorem \ref{thm0309} depends on $T=\log_{1/(1-\varepsilon)}kc_{\max}/c_{\min}$, where $c_{\max}/c_{\min}$ could be arbitrarily large.  To this end, we add a preprocessing step to Algorithm \ref{algo1} in order to create a modified instance with bounded $c_{\max}/c_{\min}$. The complete algorithm is presented in Algorithm \ref{algo4}. We first sort all elements in non-decreasing cost such that $c(v_1)\leq c(v_2)\leq \ldots\leq c(v_m)$. Then we compute the the minimum $j$ such that $g(\{v_1,\ldots,v_j\})\geq k$. Notice that $\{v_1,\ldots,v_j\}$ must be a feasible solution to our problem. Let $V_0\leftarrow \{v\in V\colon c(v)<\frac{\varepsilon}{mk}c(v_j)\}$ and $V_1\leftarrow \{v\in V\colon c(v)>jc(v_j)\}$. That is, $V_0$ contains all elements with low cost and $V_1$ contains all elements with high cost. Let $V^{mod}\leftarrow V-(V_0\cup V_1)$ denote the set of elements with ``moderate'' cost. Then we apply Algorithm \ref{algo1} to $V^{mod}$ to obtain an output $B^{mod}$. Because $V^{mod}$ contains all elements with moderate cost, we can bound the ratio $c_{\max}/c_{\min}$ as $c_{\max}/c_{\min}\leq kmj/\varepsilon$, where we abuse the notations to use $c_{\max}$ and $c_{\min}$ to denote the highest and lowest cost in $V^{mod}$ respectively. At last, $B^{mod}\cup V_0$ is returned as the final solution. 
We next present the main theorem of this paper.
\begin{theorem}\label{thm1018}
With probability at least $1-3\varepsilon$, for any $0<\varepsilon< 1/5$, Algorithm \ref{algo4} achieves an approximation ratio of at most $\frac{H(\min\{\Delta,k\})}{1-5\varepsilon}$. It takes $O(\frac{\log km\log k(\log m+\log\log mk)}{\varepsilon^4})$  rounds.
\end{theorem}
The proof of Theorem \ref{thm1018} is given in Appendix J.
\begin{algorithm}[H]
\caption{MinSMC-Main}\label{algo4}
\textbf{Input:} MinSMC instance $\mathcal I=(V,g,c,k)$ and a constant $0<\varepsilon<1/4$.

\textbf{Output:} A subset $V'\subseteq V$ such that $g(V')\geq k$.
	
\begin{algorithmic}[1]
\STATE index all elements in increasing order of costs
\STATE $j\leftarrow \arg\min\{i\colon g(\{v_1,\ldots,v_i\})\geq k\}$
\STATE $V_0\leftarrow \{v\in V\colon c(v)<\frac{\varepsilon}{mk}c(v_j)\}$
\STATE $V_1\leftarrow \{v\in V\colon c(v)>jc(v_j)\}$
\STATE $V^{mod}\leftarrow V-(V_0\cup V_1)$
\STATE $g^{mod}\leftarrow g_{V_0}$ where $g_{V_0}$ is the marginal profit function of the set over $V_0$
\STATE $k^{mod}\leftarrow \max\{0,k-g(V_0)\}$
\STATE Let $ \mathcal I^{mod}=(V^{mod},g^{mod},c,k^{mod},\varepsilon)$
\STATE $B^{mod}\leftarrow $MinSMC-Par($\mathcal I^{mod}$)
\STATE $V'\leftarrow B^{mod}\cup V_0$
\end{algorithmic}
\end{algorithm}


\section{Conclusion and Discussion}\label{sec3}
In this paper, we present a parallel algorithm for the MinSMC to obtain a solution that achieves an  approximation ratio of at most $\frac{H(\min\{\Delta,k\})}{1-5\varepsilon}$, with probability at least $1-3\varepsilon$, in $O(\frac{\log km\log k(\log m+\log\log mk)}{\varepsilon^4})$ rounds, where $0<\varepsilon<1/5$ is a constant. How to obtain a near  $H(\min\{\Delta,k\})$-approximation  parallel algorithm using less number of rounds is a topic deserving further exploration.

\acks{This research is supported in part by National Natural Science Foundation of China (11901533, U20A2068, 11771013), and Zhejiang Provincial Natural Science Foundation of China (LD19A010001). }


\bibliography{reference}

\newpage
\appendix

\section{Proof of Lemma \ref{lem0303-2}}

\begin{proof}
Assume $|A|=a$, 
\begin{align*}
\ & \ \mathbb{E}[I_{t'}(X',x')]= \sum\limits_{X'=\{x_1,\ldots,x_{t'}\},x'}I[g_{B\cup\{x_1,\ldots,x_{t'}\}}(x')\geq (1-\varepsilon)\tau]\Pr[\textrm{$x_1,\ldots,x_{t'},x'$ is picked}]\\ \nonumber
=\ & \ \frac{1}{a\times (a-1)\times\cdots \times (a-t')}\sum\limits_{X'=\{x_1,\ldots,x_{t'}\},x'}I[g_{B\cup\{x_1,\ldots,x_{t'}\}}(x')\geq(1-\varepsilon)\tau]\\ \nonumber\leq \ & \ \frac{1}{a\times \cdots\times(a-t')}\sum\limits_{X'=\{x_1,\ldots,x_{t'}\},x'}I[g_{B\cup\{x_1,\ldots,x_{t}\}}(x')\geq(1-\varepsilon)\tau]\\
\nonumber = \ & \ \frac{[a-(t+1)]\times[a-(t+2)]\times\cdots\times (a-t')}{a\times \cdots\times(a-t')}\sum\limits_{X=\{x_1,\ldots,x_t\},x'}I[g_{B\cup\{x_1,\ldots,x_{t}\}}(x')\geq(1-\varepsilon)\tau]\\
\nonumber = \ & \ \frac{1}{a(a-1)\cdots(a-t)}\sum\limits_{X=\{x_1,\ldots,x_t\},x'}I[g_{B\cup\{x_1,\ldots,x_{t}\}}(x')\geq(1-\varepsilon)\tau]\\
\nonumber \leq \ & \ \mathbb{E}[I_t(X,x)],
\end{align*}
where the first inequality is due to the submodularity of function $g$, and the last inequality holds because $x$ is sampled from $A\setminus X\supseteq A\setminus X'$ and function $I$ is nonnegative.
\end{proof}

\section{Proof of Lemma \ref{lem0302-1}}

\begin{proof}
Let $Y_{m'}=\sum_{i=1}^{m'}I_{t_p,B_p,A_p,\tau,\epsilon}(X_i,x_i)$, and let $\mu=\mathbb{E}[I_{t_p,B_p,A_p,\tau,\epsilon}(X,x)]$. By the Chenorff bound (see \citep{Mitzenmacher}), for any $a>0$,
\begin{equation}\label{eq0302-2}
\Pr[|Y_{m'}-{m'}\mu|\geq a]\leq 2e^{-\frac{a^2}{2{m'}\mu}}.
\end{equation}
For $a=\frac{\bar{\varepsilon} m'}{2}$, using $m'=8\lceil\log(2/\delta)/\bar{\varepsilon}^2\rceil$ and $\mu\leq 1$, we have
\begin{equation}\label{eq0625}
\frac{a^2}{2m'\mu}\geq \log(2/\delta).
\end{equation}
Combining inequalities \eqref{eq0302-2} and \eqref{eq0625}, we have $\Pr[|Y_{m'}-{m'}\mu|\geq \frac{\bar{\varepsilon} m'}{2}]\leq \delta$. That is, $\Pr[|\bar{\mu}_p-\mu|\geq \frac{\bar{\varepsilon}}{2}]\leq \delta$. If $\bar{\mu}_p\leq 1-1.5\bar{\varepsilon}$, then $\Pr(\mu >1-\bar{\varepsilon})\leq \Pr(|\bar{\mu}_p-\mu|\geq {\bar{\varepsilon}}/2)\leq \delta$, that is, with probability at least $1-\delta$, we have $\mathbb{E}[I_{t_p,B_p,A_p,\tau,\epsilon}(X,x)]\leq 1-\bar{\varepsilon}$. The second half of the lemma can be proved similarly.
\end{proof}

\section{Proof of Lemma \ref{lem0308-1}}

\begin{proof}
For a $p\leq r$, if the for loop is executed $i=\log_{1+\bar{\varepsilon}}m$ rounds, then $t_p=|A_p|$, and $A_{p+1}$ becomes empty. Next, consider the case when the for loop is exited because of line \ref{line18}. Denote $C_p$ to be the event $\mathbb{E}[|A_{p+1}|]\leq (1-\varepsilon)|A_p|$. By Lemma \ref{lem0302}, $\Pr(\bar{C_p})\leq \delta=\varepsilon/(2rkT\ell)$. By the union bound,
\begin{align*}
\Pr[C_1\cap C_2\cap \ldots \cap C_r]=\ & \ 1-\Pr[\bar{C_1}\cup \bar{C_2}\cup \ldots\cup \bar{C_r}]\\ \nonumber
\geq \ & \ 1-\sum_{i=1}^{r}\Pr(\bar{C_i})\\ \nonumber
\geq \ & \ 1-\varepsilon/(2kT^2\ell).
\end{align*}
So, with probability at least $1-\varepsilon/(2nT^2\ell)$, we have
$$
\mathbb{E}[|A_r|]\leq (1-\bar{\varepsilon})^r\cdot \mathbb{E}[|A_1|]\leq \varepsilon/(2T\ell).
$$
Denote the event $C'$ to be $\mathbb{E}[|A_r|]\leq \varepsilon/(2T\ell)$. We have proved $\Pr(C')\geq 1-\varepsilon/(2kT^2\ell)$. Using Markov's inequality, $\Pr(|A_r|\geq 1|C')\leq \varepsilon/(2T^2\ell)$. So, $\Pr(A_r=\emptyset)\geq \Pr(C')\Pr(A_r=\emptyset|C')=\Pr(C')\left(1-\Pr(|A_r|\geq1 |C')\right)\geq (1-\varepsilon/(2T^2k\ell))(1-\varepsilon/(2T^2\ell))\geq 1-\varepsilon/(T^2\ell)$. The lemma is proved.
\end{proof}

\section{Proof of Corollary \ref{cor0308}}

\begin{proof}
For $1\leq t'\leq \ell$, let $C_t^{t'}$ be the event of $A_{t}^{t'}=\emptyset$. Lemma \ref{lem0308-1} says that $\Pr(C_t^{t'})\geq 1-\varepsilon/(T\ell)$ after $J_{t}^{t'}$ is computed (line \ref{line0704-1} of Algorithm \ref{algo1}). Hence,
\begin{align*}
\Pr(A_t=\emptyset)=\ & \ \Pr(C_t^1\cap C_t^2\ldots\cap C_t^{\ell})\\ \nonumber
=\ & \ 1-\Pr(\bar{C_t^1}\cup \bar{C_t^2}\cup \ldots\cup \bar{C_t^\ell})\\ \nonumber
\geq \ & \ 1-\sum_{i=1}^{\ell}\Pr(\bar{C_t^i})\\ \nonumber
\geq \ & \ 1-\varepsilon/T^2.
\end{align*}
The lemma is proved.
\end{proof}

\section{Proof of Corollary \ref{cor0303}}

\begin{proof}
By Lemma \ref{lem0303} and the union bound, with probability at least $1-\delta nr\ell$,
\begin{equation}\label{eq0706-1}
\sum_{t'=1}^{\ell}\mathbb{E}[g_{B_{t}^{t'}}(J_{t}^{t'})]\geq (1-\varepsilon)^2\sum_{t'=1}^{\ell}\sum_{v\in J_{t}^{t'}}g_{B_{t}^{t'}}(v).
\end{equation}
By the definition of $A_t^{t'}$ (line \ref{line11} of Algorithm \ref{algo1}), for any $t'\leq \ell$ and $v\in J_{t}^{t'}$, we have $\frac{g_{B_{t}^{t'}}(v)}{c(v)}\geq (1-\varepsilon)^{t}\beta$. Then by inequality \eqref{eq0706-1}, with probability at least $1-\delta nr\ell$,
\begin{align*}
\frac{\mathbb{E}[g_{B_{t}^1}(D_t)]}{c(D_t)}=\ & \ \frac{\sum_{t'=1}^{\ell}\mathbb{E}[g_{B_{t}^{t'}}(J_{t}^{t'})]}{c(D_t)}\\ \nonumber
\geq \ & \ \frac{(1-\varepsilon)^2\sum_{t'=1}^{\ell}\sum_{v\in J_{t}^{t'}}g_{B_{t}^{t'}}(v)}{c(D_t)}\\ \nonumber
=\ & \ \frac{(1-\varepsilon)^2\sum_{t'=1}^{\ell}\sum_{v\in J_{t}^{t'}}g_{B_{t}^{t'}}(v)}{\sum_{t'=1}^{\ell}\sum_{v\in J_{t}^{t'}}c(v)}\\ \nonumber
\geq\ & \ (1-\varepsilon)^{t+2}\beta.
\end{align*}
The lemma is proved.
\end{proof}

\section{Proof of Claim 1 in Theorem \ref{thm0309}}

\begin{proof}
If $D_1,D_2,\ldots,D_T$ satisfy \eqref{eq0712-1} for all $i=1,\ldots,T$, then the claim holds by letting $p=T$ and $D_i'=D_i$ for $i=1,\ldots,T$. Otherwise, let $t$ be the minimum index with
\begin{equation}\label{eq0713}
\frac{\mathbb{E}[g_{B_{t}^1}(D_{t})]}{c(D_{t})}>\frac{\mathbb{E}[g_{B_{t-1}^1}(D_{t-1})]}{c(D_{t-1})}.
\end{equation}
Let $D'_1=D_1,D'_2=D_2, \ldots, D'_{t-3}=D_{t-3}$. Inequality \eqref{eq0712-1} holds for $i=1,\ldots,t-4$. We next focus on constructing $D_i'$ for $i\geq t-2$ adaptively, starting from $i=t-2$. Note that for $i\geq t-2$ , $D_i'$ may contain multiple sets from $\{D_j\colon t-2\leq j\leq T\}$.

\vskip 0.2cm {\bf Some preparations.} We first prove that with probability at least $1-\frac{\varepsilon}{2T^2}$, 
\begin{equation}\label{eq0308-4}
\frac{\mathbb{E}[g_{B_{t-1}^1}(D_{t-1}\cup D_t)]}{c(D_{t-1}\cup D_t)}\geq (1-\varepsilon)^{t+1}\beta,
\end{equation}
and with probability at least $1-\frac{\varepsilon}{T}$,
\begin{equation}\label{eq0308-4-2}
\frac{\mathbb{E}[g_{B_{t-1}^1}(D_{t-1}\cup D_t)]}{c(D_{t-1}\cup D_t)}\leq (1-\varepsilon)^{t-1}\beta.
\end{equation}
In fact, by Corollary \ref{cor0303}, with probability at least $1-\delta nr\ell =1-\frac{\varepsilon}{2T^2}$,
\begin{equation}\label{eq0713-1}
\frac{\mathbb{E}[g_{B_{t-1}^1}(D_{t-1})]}{c(D_{t-1})}\geq (1-\varepsilon)^{t+1}\beta.
\end{equation}
Observe that
\begin{equation}\label{eq0715-1}
\frac{\mathbb{E}[g_{B_{t-1}^1}(D_{t-1}\cup D_t)]}{c(D_{t-1}\cup D_t)} =\frac{\mathbb{E}[g_{B_t^1}(D_t)]+\mathbb{E}[g_{B_{t-1}^1}(D_{t-1})]}{c(D_t)+c(D_{t-1})}.
\end{equation}
\eqref{eq0308-4} follows from inequalities \eqref{eq0713}, \eqref{eq0713-1}, \eqref{eq0715-1}, and the observation that $B_{t}^1=B_{t-1}^1\cup D_{t-1}$.
By Corollary \ref{cor0308}, after $J_{t-1}^{\ell}$ is computed,  the primary bucket $A_{t-1}$ becomes empty with probability at least $1-\varepsilon/T^2$. By the union bound, with probability at least $1-\varepsilon/T$, all primary buckets $A_i$ with $i\leq t-1$ become empty, and this implies that
\begin{equation}\label{eq0802-3}
\mbox{every remaining element $v$ satisfies $g_{B_t^{1}}(v)/c(v)<(1-\varepsilon)^{t-1}\beta$}.
\end{equation}

Thus, by the submodularity of function $g_{B_t^1}$, we have $g_{B_t^1}(D_t)\leq\sum_{v\in D_t}g_{B_t^1}(v)<(1-\varepsilon)^{t-1}\beta c(D_t)$. Combining this with inequality \eqref{eq0713}, both $\frac{\mathbb{E}[g_{B_{t}^1}(D_{t})]}{c(D_{t})}$ and $\frac{\mathbb{E}[g_{B_{t-1}^1}(D_{t-1})]}{c(D_{t-1})}$ are upper bounded by $(1-\varepsilon)^{t-1}\beta$ with probability at least $1-\varepsilon/T$. This, together with \eqref{eq0715-1}, implies \eqref{eq0308-4-2}.

Similar to the proofs of \eqref{eq0308-4} and \eqref{eq0308-4-2}, we can prove that each of the following statements holds with probability at least $1-\frac{3\varepsilon}{2T}$:

If $\frac{\mathbb{E}[g_{B_{t-2}^1}(D_{t-2})]}{c(D_{t-2})}\geq (1-\varepsilon)^{t-1}\beta$ and $\frac{\mathbb{E}[g_{B_{t-1}^1}(D_{t-1}\cup D_t)]}{c(D_{t-1}\cup D_t)}< \frac{\mathbb{E}[g_{B_{t+1}^1}(D_{t+1})]}{c(D_{t+1})}$, then
\begin{equation}\label{eq0801-1}
(1-\varepsilon)^{t+1}\beta\leq \frac{\mathbb{E}[g_{B_{t-1}^1}(D_{t-1}\cup D_t\cup D_{t+1})]}{c(D_{t-1}\cup D_t\cup D_{t+1})}\leq (1-\varepsilon)^{t-1}\beta.
\end{equation}

If $\frac{\mathbb{E}[g_{B_{t-3}'}(D_{t-2})]}{c(D_{t-2})}< (1-\varepsilon)^{t-1}\beta$ and $\frac{\mathbb{E}[g_{B_{t-2}^1}(D_{t-2}\cup D_{t-1}\cup D_t)]}{c(D_{t-2}\cup D_{t-1}\cup D_t)}\geq \frac{\mathbb{E}[g_{B_{t+1}^1}(D_{t+1})]}{c(D_{t+1})}$, then
\begin{equation}\label{eq0801-2}
(1-\varepsilon)^{t+1}\beta\leq \frac{\mathbb{E}[g_{B_{t-2}^1}(D_{t-2}\cup D_{t-1}\cup D_t)]}{c(D_{t-2}\cup D_{t-1}\cup D_t)}\leq (1-\varepsilon)^{t-1}\beta.
\end{equation}

If $\frac{\mathbb{E}[g_{B_{t-3}'}(D_{t-2})]}{c(D_{t-2})}< (1-\varepsilon)^{t-1}\beta$ and $\frac{\mathbb{E}[g_{B_{t-2}^1}(D_{t-2}\cup D_{t-1}\cup D_t)]}{c(D_{t-2}\cup D_{t-1}\cup D_t)}< \frac{\mathbb{E}[g_{B_{t+1}^1}(D_{t+1})]}{c(D_{t+1})}$, then
\begin{equation}\label{eq0801-3}
(1-\varepsilon)^{t+1}\beta\leq \frac{\mathbb{E}[g_{B_{t-2}^1}(D_{t-2}\cup D_{t-1}\cup D_t\cup D_{t+1})]}{c(D_{t-2}\cup D_{t-1}\cup D_t\cup D_{t+1})}\leq (1-\varepsilon)^{t-1}\beta.
\end{equation}

Observe that in each of (\ref{eq0801-1}),(\ref{eq0801-2}),(\ref{eq0801-3}), the lower bound holds with probability at least $1-\frac{\varepsilon}{2T^2}$ and the upper bound holds with probability at least $1-\frac{\varepsilon}{T}$. Hence, by the union bound, the probability that each of (\ref{eq0801-1}),(\ref{eq0801-2}),(\ref{eq0801-3}) holds is at least $1-(\frac{\varepsilon}{2T^2}+\frac{\varepsilon}{T})\geq 1-\frac{3\varepsilon}{2T}$.

\vskip 0.2cm {\bf Construction of $D_{t-2}'$.} Next, we show how to construct $D_{t-2}'$ such that it satisfies property \eqref{eq0712-1} and the following two properties:

($\romannumeral1$) If  $D'_{t-2}=D_{t-2}$, then at the same time we can construct $D'_{t-1}$ satisfying that it contains more than one sets from $\{D_j\colon 1\leq j\leq T\}$, and with probability at least $1-\frac{3\varepsilon}{2T}$,
\begin{equation}\label{eq0730-1}
(1-\varepsilon)^{t+1}\beta \leq \frac{\mathbb{E}[g_{B_{t-2}'}(D'_{t-1})]}{c(D'_{t-1})}\leq (1-\varepsilon)^{t-1}\beta.
\end{equation}
Else $D'_{t-2}$ contains more than one sets from $\{D_j\colon 1\leq j\leq T\}$, and with probability at least $1-\frac{3\varepsilon}{2T}$,
\begin{equation}\label{eq0724-9}
(1-\varepsilon)^{t+1}\beta\leq \frac{\mathbb{E}[g_{B_{t-3}'}(D'_{t-2})]}{c(D'_{t-2})}\leq (1-\varepsilon)^{t-1}\beta.
\end{equation}
Furthermore, for both inequalities \eqref{eq0730-1} and \eqref{eq0724-9}, the lower bound $(1-\varepsilon)^{t+1}\beta$ holds with probability at least $1-\frac{\varepsilon}{2T^2}$.


($\romannumeral2$) Suppose $D'_i=D_j\cup D_{j+1}\cup \ldots\cup D_{j'}$, then with probability at least $1-\frac{3\varepsilon}{2T}$,
\begin{equation}\label{eq0730-2}
\frac{\mathbb{E}[g_{B_{i-1}'}(D'_{i})]}{c(D'_{i})}\geq \frac{\mathbb{E}[g_{B_{j'+1}^1}(D_{j'+1})]}{c(D_{j'+1})}.
\end{equation}

To complete the construction, we consider two cases.

\vskip 0.2cm {\bf Case 1.} $\frac{\mathbb{E}[g_{B_{t-3}'}(D_{t-2})]}{c(D_{t-2})}\geq (1-\varepsilon)^{t-1}\beta$.

In this case, let $D'_{t-2}=D_{t-2}$. Then \eqref{eq0712-1} holds for $i=t-3$ by the choice of $t$. Furthermore, by \eqref{eq0308-4-2} and the condition of Case 1, with probability at least $1-\frac{\varepsilon}{T}$, we have
\begin{equation}\label{eq0724-1}
\frac{\mathbb{E}[g_{B_{t-3}'}(D'_{t-2})]}{c(D'_{t-2})}\geq \frac{\mathbb{E}[g_{B_{t-2}'}(D_{t-1}\cup D_t)]}{c(D_{t-1}\cup D_t)}.
\end{equation}

In this case, we shall construct $D_{t-1}'$ by distinguishing two subcases.

{\bf Subcase 1.1}$\frac{\mathbb{E}[g_{B_{t-2}'}(D_{t-1}\cup D_t)]}{c(D_{t-1}\cup D_t)}\geq \frac{\mathbb{E}[g_{B_{t+1}^1}(D_{t+1})]}{c(D_{t+1})}$.

In this subcase, let $D'_{t-1}=D_{t-1}\cup D_t$. By \eqref{eq0724-1}, with probability at least $1-\varepsilon/T$,
$$\frac{\mathbb{E}[g_{B_{t-3}'}(D'_{t-2})]}{c(D'_{t-2})}\geq \frac{\mathbb{E}[g_{B_{t-2}'}(D'_{t-1})]}{c(D'_{t-1})},$$
which satisfies \eqref{eq0712-1} for $i=t-2$. By the condition of Subcase 1.1,
$$\frac{\mathbb{E}[g_{B_{t-2}'}(D'_{t-1})]}{c(D'_{t-1})}\geq \frac{\mathbb{E}[g_{B_{t+1}^1}(D_{t+1})]}{c(D_{t+1})},$$
which satisfies \eqref{eq0730-2} for $i=t-1$. Furthermore, \eqref{eq0730-1} follows from \eqref{eq0308-4} and \eqref{eq0308-4-2}.

\vskip 0.2cm {\bf Subcase 1.2} $\frac{\mathbb{E}[g_{B_{t-2}'}(D_{t-1}\cup D_t)]}{c(D_{t-1}\cup D_t)}< \frac{\mathbb{E}[g_{B_{t+1}^1}(D_{t+1})]}{c(D_{t+1})}$.

In this subcase, let $D'_{t-1}=D_{t-1}\cup D_t\cup D_{t+1}$. Then \eqref{eq0730-1} holds by \eqref{eq0801-1}.  Combining the right hand side of \eqref{eq0801-1} with the condition of Case 1, with probability at least $1-\frac{\varepsilon}{T}$,
$$ \frac{\mathbb{E}[g_{B_{t-3}'}(D'_{t-2})]}{c(D'_{t-2})}\geq\frac{\mathbb{E}[g_{B_{t-2}'}(D'_{t-1})]}{c(D'_{t-1})},$$
which satisfies \eqref{eq0712-1} for $i=t-2$. Similar to the proof of inequality \eqref{eq0308-4-2}, under the condition of Subcase 1.2, with probability at least $1-\frac{\varepsilon}{T}$,
\begin{equation}\label{eq0724-5}
\frac{\mathbb{E}[g_{B_{t+2}^1}(D_{t+2})]}{c(D_{t+2})}\leq (1-\varepsilon)^{t+1}\beta.
\end{equation}
Combining inequality \eqref{eq0724-5} with the left hand side of \eqref{eq0801-1}, by the union bound, with probability at least $1-\frac{3\varepsilon}{2T}$,
\begin{equation}\label{eq0725}
\frac{\mathbb{E}[g_{B_{t-2}'}(D'_{t-1})]}{c(D'_{t-1})}\geq \frac{\mathbb{E}[g_{B_{t+2}^1}(D_{t+2})]}{c(D_{t+2})},
\end{equation}
which satisfies \eqref{eq0730-2} for $i=t-1$.

\vskip 0.2cm {\bf Case 2.} $\frac{\mathbb{E}[g_{B_{t-3}'}(D_{t-2})]}{c(D_{t-2})}< (1-\varepsilon)^{t-1}\beta$.


We further distinguish two subcases.

\vskip 0.2cm {\bf Subcase 2.1} $\frac{\mathbb{E}[g_{B_{t-3}'}(D_{t-2}\cup D_{t-1}\cup D_t)]}{c(D_{t-2}\cup D_{t-1}\cup D_t)}\geq \frac{\mathbb{E}[g_{B_{t+1}^1}(D_{t+1})]}{c(D_{t+1})}$.

In this subcase, let $D'_{t-2}=D_{t-2}\cup D_{t-1}\cup D_t$. Then \eqref{eq0730-2} for $i\leq t-2$ follows from the condition of this subcase, and \eqref{eq0724-9} follows from \eqref{eq0801-2}. By Corollary \ref{cor0303}, with probability at least $1-\frac{\varepsilon}{2T^2}$,
\begin{equation}\label{eq0724-7}
\frac{\mathbb{E}[g_{B_{t-4}}(D'_{t-3})]}{c(D'_{t-3})}\geq (1-\varepsilon)^{t-1}\beta.
\end{equation}
Combining \eqref{eq0724-7} with the right hand side of inequalities \eqref{eq0801-2}, by the union bound, with probability at least $1-\frac{3\varepsilon}{2T}$,
\begin{equation}\label{eq0726}
\frac{\mathbb{E}[g_{B_{t-4}}(D'_{t-3})]}{c(D'_{t-3})}\geq \frac{\mathbb{E}[g_{B_{t-3}}(D'_{t-2})]}{c(D'_{t-2})},
\end{equation}
and thus \eqref{eq0712-1} holds for $i=t-3$.

\vskip 0.2cm {\bf Subcase 2.2} $\frac{\mathbb{E}[g_{B_{t-3}'}(D_{t-2}\cup D_{t-1}\cup D_t)]}{c(D_{t-2}\cup D_{t-1}\cup D_t)}< \frac{\mathbb{E}[g_{B_{t+1}^1}(D_{t+1})]}{c(D_{t+1})}$.

In this subcase, let $D'_{t-2}=D_{t-2}\cup D_{t-1}\cup D_t\cup D_{t+1}$. Then \eqref{eq0724-9} follows from \eqref{eq0801-3}. Combining inequality \eqref{eq0724-5} with the left hand side of inequality \eqref{eq0801-3}, we have \eqref{eq0730-2} for $i\leq t-2$. Combining \eqref{eq0724-7} with the right hand side of inequality \eqref{eq0801-3}, with probability at least $1-\frac{3\varepsilon}{2T}$, \eqref{eq0712-1} holds for $i=t-3$.

The construction for $D_{t-2}'$ is completed.

\vskip 0.2cm {\bf Construction of the $D_i'$s for $i>t-2$.} Once $D_{t-2}'$ is constructed, we follow a similar procedure to construct $D_{t-1}'$. That is, we  find the next index satisfying inequality \eqref{eq0713} and construct $D_{t-1}'$ according to the same approach used to construct $D_{t-2}'$. This procedure iterates until all $D_i'$s are constructed adaptively.
%
%
%

Claim 1 is proved.
\end{proof}

\section{Proof of Claim 2 in Theorem \ref{thm0309}}

\begin{proof}
Assume $D'_i=D_j\cup\cdots\cup D_{j'}$ with $j'\geq j$. Note that $B'_{i-1}=B_j^1$. If $D'_i=D_j$, then by Corollary \ref{cor0303}, with probability at least $1-\delta nr\ell =1-\frac{\varepsilon}{2T^2}$,
\begin{equation}\label{eq0726-3} \frac{\mathbb{E}[g_{B_{i-1}'}(D'_{i})]}{c(D'_{i})}=\frac{\mathbb{E}[g_{B_j^1}(D_j)]}{c(D_j)}\geq (1-\varepsilon)^{j+2}\beta.
\end{equation}
According to Claim 1, there are four possible ways of constructing $D'_{i_t}$: \begin{align*}
& D'_{i_t}=D_{t-1}\cup D_{t},\\
& D'_{i_t}=D_{t-2}\cup D_{t-1}\cup D_t,\\
& D'_{i_t}=D_{t-2}\cup D_{t-1}\cup D_t\cup D_{t+1},\\
& D'_{i_t}=D_{t-1}\cup D_{t}\cup D_{t+1}.
\end{align*}
Hence, a $D'_i=D_j\cup\cdots\cup D_{j'}$ with $j'>j$ must be in the form of $D'_i=D'_{i_{(j+1)}}$ or $D'_i=D'_{i_{(j+2)}}$. By $(\romannumeral1)$, with probability at least $1-\frac{\varepsilon}{2T^2}$,
$$
\frac{\mathbb{E}[g_{B_{i-1}'}(D'_{i})]}{c(D'_{i})} \geq (1-\varepsilon)^{j+2}\beta \ \mbox{or} \ (1-\varepsilon)^{j+3}\beta.
$$
Combining this with inequality \eqref{eq0726-3}, 
with probability at least $1-\frac{\varepsilon}{2T^2}$,
\begin{equation}\label{eq0726-4}
\frac{\mathbb{E}[g_{B_{i-1}'}(D'_{i})]}{c(D'_{i})}\geq (1-\varepsilon)^{j+3}\beta.
\end{equation}
Similar to the proof of \eqref{eq0802-3}, after $D_1,\ldots,D_{j-1}$ are chosen, for any remaining element $v$, with probability at least $1-\varepsilon/T$, $\mathbb{E}[g_{B_j^1}(v)]/c(v)\leq (1-\varepsilon)^{j-1}\beta$. Hence
$$
\beta(B_{i-1}')=\beta(B_j^1)\leq (1-\varepsilon)^{j-1}\beta.
$$
Combining this with inequality \eqref{eq0726-4}, by the union bound, with probability at least $1-3\varepsilon/2T$,
a {\em fixed} $D'_i$ satisfies Claim 2. Again by the union bound, for {\em every} $1\leq i\leq p-1$, Claim 2 holds with probability at least $1-3\varepsilon/2$.
\end{proof}

\section{Proof of Claim 3 in Theorem \ref{thm0309}}

\begin{proof}
By the definition of $w(v)$, we have
\begin{align}\label{eq0328-5}
w(A^*)=&\sum_{v\in A^*}\sum_{i=1}^p(z_{v,i}-z_{v,i+1})\frac{c(D_i')}{r_i} \nonumber \\
=&\frac{c(D_1')}{r_1}\sum\limits_{v\in A^*}z_{v,1}+
\sum\limits_{i=2}^p \left(\frac{c(D'_i)}{r_i}-\frac{c(D'_{i-1})}{r_{i-1}}\right)\sum\limits_{v\in A^*}z_{v,i}.
\end{align}
Similarly, $c(B_p')$ can be rewritten as follows:
\begin{align}
c(B_p')= &  \sum_{i=1}^pc(D_i')=\sum_{i=1}^pr_i\cdot\frac{c(D'_i)}{r_i}\nonumber\\
= & \sum\limits_{i=1}^p \left(\sum\limits_{j=i}^p r_j-\sum\limits_{j=i+1}^p r_j\right)\frac{c(D'_i)}{r_i} \nonumber \\
= & \frac{c(D'_1)}{r_1}\sum\limits_{j=1}^p r_j+\sum\limits_{i=2}^p\left(\frac{c(D'_i)}{r_i}-\frac{c(D'_{i-1})}{r_{i-1}}\right)\sum\limits_{j=i}^p r_j\label{eq0328-6}
\end{align}
By Claim 1, with probability at least $1-3\varepsilon/2$,
\begin{equation}\label{eq0706-3}
\frac{c(D'_i)}{r_i}\geq \frac{c(D'_{i-1})}{r_{i-1}}\ \mbox{holds for any $2\leq i\leq p$}.
\end{equation}
Then, comparing \eqref{eq0328-5} and \eqref{eq0328-6}, to prove Claim 3, it suffices to prove that
\begin{equation}\label{eq0328-7}
\sum\limits_{j=i}^pr_j\leq \sum\limits_{v\in A^*}z_{v,i}\ \mbox{for any $1\leq i\leq p$}.
\end{equation}

The left term of inequality \eqref{eq0328-7} can be written as
\begin{equation}\label{eq0328-1}
\sum\limits_{j=i}^pr_j=\sum\limits_{j=i}^p \left(\mathbb{E}[g(B_{j}')]-\mathbb{E}[g(B_{j-1}')]\right)
=\mathbb{E}[g(B_{p}')]-\mathbb{E}[g(B_{i-1}')]=k-\mathbb{E}[g(B_{i-1}')].
\end{equation}
Suppose $A^*=\{v_1,\ldots,v_q\}$. Denote $A_j^*=\{v_1,\ldots,v_j\}$ for $j=1,\ldots,q$, and $A_0^*=\emptyset$. The right term of inequality \eqref{eq0328-7} can be bounded by
\begin{align}\label{eq0328-8}
\sum\limits_{v\in A^*}z_{v,i}&=\sum_{j=1}^q\mathbb{E}[g_{B_{i-1}'}(v_j)] \nonumber \\
&\geq \sum_{j=1}^q\mathbb{E}[g_{B_{i-1}'\cup A_{j-1}^*}(v_j)] \nonumber \\
& =\sum_{j=1}^q\big(\mathbb{E}[g(B_{i-1}'\cup A_j^*)]-\mathbb{E}[g(B_{i-1}'\cup A_{j-1}^*)]\big)\nonumber \\
&=\mathbb{E}[g(B_{i-1}'\cup A^*)]-\mathbb{E}[g(B_{i-1}')]=k-\mathbb{E}[g(B_{i-1}')],
\end{align}
where the inequality is due to the submodularity of $g$. Inequality \eqref{eq0328-7} follows from \eqref{eq0328-1} and \eqref{eq0328-8}, and thus Claim 3 is proved.
\end{proof}

\section{Proof of Claim 4 in Theorem \ref{thm0309}}

\begin{proof}
By Claim 2, for any $v\in V$ and any $1\leq i\leq p-1$, with probability at least $1-3\varepsilon/2$,
\begin{equation}\label{eq0309-1}
(1-\varepsilon)^4\frac{\mathbb{E}[g_{B_{i}'}(v)]}{c(v)}\leq \frac{\mathbb{E}[g_{B_{i}'}(D'_{i+1})]}{c(D'_{i+1})}.
\end{equation}
It follows that
\begin{align}\label{eq18-03-21-4}
w(v)
=& \sum\limits_{i=1}^p(z_{v,i}-z_{v,i+1})\frac{c(D_i')}{r_i}  \nonumber\\
\leq & \sum\limits_{i=1}^{p}(z_{v,i}-z_{v,i+1})\frac{c(v)}{(1-\varepsilon)^4z_{v,i}}\nonumber\\
\leq & c(v)\cdot H(z_{v,1})/(1-4\varepsilon).
\end{align}
Then Claim 4 follows from $z_{v,1}\leq \max_{v\in V}g(v)\leq \min\{\Delta,k\}$ and $(1-\varepsilon)^4\geq 1-4\varepsilon$.
\end{proof}

\section{Proof of Theorem \ref{thm1018}}

\begin{proof}

Suppose $A^*$ is an optimal solution of our original problem. Recall that
\[j=\arg\min\{i\colon g(\{v_1,\ldots,v_i\})\geq k\}\] for ordered elements $c(v_1)\leq c(v_2)\leq \cdots \leq c(V_m)$. We first show that
\begin{equation}\label{eq0224}
c(v_j)\leq c(A^*).
\end{equation}
To prove this, suppose $c(v_{j'})=\max\{c(v_i)\colon v_i\in A^*\}$. If $c(v_{j'})<c(v_j)$, then $A^*\subseteq \{v_1,\ldots,v_{j'}\}$. By the monotonicity of $g$, we have $k=g(A^*)\leq g(\{v_1,\ldots,v_{j'}\})$, contradicting to the definition of $j$. So, $c(v_j)\leq c(v_{j'})\leq c(A^*)$. On the other hand,
\begin{equation}\label{eq0224-1}
c(A^*)\leq j\cdot c(v_j).
\end{equation}
This is because $g(\{v_1\cup\ldots\cup v_j\})\geq k$, $\{v_1,\ldots,v_j\}$ is a feasible solution whose cost is $c(v_1)+\cdots+c(v_j)\leq jc_j$. Since $A^*$ has the minimum cost among all feasible solutions, \eqref{eq0224-1} follows. As a consequence of \eqref{eq0224}, we have
$$c(V_0)\leq |V_0|\frac{\varepsilon}{mk}{c(v_j)}\leq \frac{\varepsilon}{k}c(A^*),$$ where $V_0= \{v\in V\colon c(v)<\frac{\varepsilon}{mk}c(v_j)\}$.
As a consequence of \eqref{eq0224-1}, $A^*$ must not contain any elements $v$ such that $c(v)>jc(v_j)$. So, the optimal solution for the instance $(V-V_1,g,c,k)$ is the same as $A^*$, where $V_1= \{v\in V\colon c(v)>jc(v_j)\}$.

For instance $\mathcal I^{mod}$, we have $\frac{c_{\max}}{c_{\min}}\leq \frac{jc(v_j)}{\frac{\varepsilon}{mk}{c(v_j)}}=\frac{jmk}{\varepsilon}\leq \frac{m^2k}{\varepsilon}$. Substituting this bound into the expression of $T$,
by Theorem \ref{thm0309}, Algorithm \ref{algo1} produces a solution $B^{mod}$ in $O(\frac{\log km\log k(\log m+\log\log mk)}{\varepsilon^4})$ rounds such that
$$
c(B^{mod})\leq \frac{H(\Delta)}{1-4\varepsilon}opt',
$$
where $opt'$ is the optimal value for instance $\mathcal I^{mod}$. By the submodularity and the monotonicity of $g$, we have $k=g(A^*\cup V_0)\leq g(A^*\setminus V_0)+g(V_0)$, and thus $g(A^*\setminus V_0)\geq k^{mod}$, that is, $A^*\setminus V_0$ is a feasible solution to $\mathcal I^{mod}$. It follows that $opt'\leq c(A^*\setminus V_0)$, and thus
the output of $V'$ in Algorithm \ref{algo3} has cost
$$c(V')=c(B^{mod})+c(V_0)\leq \left(\frac{H(\Delta)}{1-4\varepsilon}+\frac{\varepsilon}{k}\right)c(A^*)\leq \frac{H(\Delta)}{1-5\varepsilon}c(A^*).$$
The theorem is proved.
\end{proof}

\end{document}